%% file: main.tex
\documentclass{ucetd}
\usepackage{subfigure,epsfig,amsfonts}
\usepackage{natbib}
\usepackage{float}
\usepackage{dsfont}
\usepackage{amsmath,amssymb,mathtools}
\usepackage{amsthm}
\usepackage{algorithm,algorithmic}
\usepackage[utf8]{inputenc}
\usepackage{tikz}
\usepackage{bm}

\newtheorem{claim}{Claim}[chapter]
\newtheorem{observation}{Observation}[chapter]
\newtheorem{corollary}{Corollary}[chapter]
\newtheorem{theorem}{Theorem}[chapter]
\newtheorem{lemma}{Lemma}[chapter]
\newtheorem{definition}{Definition}[chapter]

\newcommand{\OPT}{\mathsf{O}\mathsf{P}\mathsf{T}}
\newcommand{\ALG}{\mathsf{A}\mathsf{L}\mathsf{G}}
\newcommand{\LP}{\mathsf{L}\mathsf{P}}

\newcommand{\given}{\;|\;}

\newcommand{\PlusSign}{\text{``$+$''}}
\newcommand{\MinusSign}{\text{``$-$''}}
\newcommand{\ONE}{\mathds{1}}

\newcommand{\openLP}{\hrule height 0.8pt\rule{0pt}{1pt}} %\hrule height 0.4pt\rule{0pt}{1pt}}
\newcommand{\closeLP}{\rule{0pt}{1pt}\hrule height 0.4pt\rule{0pt}{1pt}} %\hrule height 0.8pt\rule{0pt}{12pt}}

\title{Four Algorithms for Correlation Clustering: A Survey}
\author{Jafar Jafarov}
\department{Computer Science}
\division{Physical Sciences}
\degree{Masters}
\date{May 2020}

\begin{document}

\maketitle

\makecopyright

\tableofcontents

\acknowledgments

\begin{flushleft}
I am deeply indebted to Prof. Yury Makarychev for introducing me to the subject and for his wise guidance and continuous support.\\
\vspace{0.2 cm}
I am grateful to Prof. Janos Simon for his continuous support and helpful advise.\\ 
\vspace{0.2 cm}
I would like to express my gratitude to Prof. Aaron Potechin for devoting his time and energy to this work and for being my committee member.\\
\vspace{0.2 cm}
Lastly, I am grateful to Aritra Sen, Tushant Mittal, Gokalp Demirci and to all my friends and colleagues in the TCS group at the University of Chicago.

\end{flushleft}

\abstract
In the Correlation Clustering problem, we are given a set of objects with pairwise similarity information. Our aim is to partition these objects into clusters that match this information as closely as possible. More specifically,  the pairwise information is given as a weighted graph $G$ with its edges labelled as ``similar" or ``dissimilar" by a binary classifier. The goal is to produce a clustering that minimizes the weight of ``disagreements": the sum of the weights of similar edges across clusters and dissimilar edges within clusters.

In this exposition we focus on the case when $G$ is complete and unweighted. We explore four approximation algorithms for the Correlation Clustering problem under this assumption. In particular, we describe the following algorithms: \textbf{(i)} the $17429-$approximation algorithm by~\citet*{BBC04}, \textbf{(ii)} the $4-$approximation algorithm by~\citet*{CGW03} \textbf{(iii)} the $3-$approximation algorithm by~\citet*{ACN08} \textbf{(iv)} the $2.06-$approximation algorithm by~\citet*{CMSY15}. 
\mainmatter

\input{Introduction}
\input{Notations}
\input{Bansal-Blum-Chawla-Algorithm}
\input{Charikar-Guruswami-Wirth-Algorithm}
\input{Ailon-Charikar-Newman-Algorithm}
\input{Chawla-Makarychev-Schramm-Yaroslavtsev-Algorithm}

\makebibliography

%\bibliographystyle{plainnat}
%\bibliography{corr-clust}

\end{document}

%% file: Introduction.tex
\chapter{Introduction}
In the Correlation Clustering problem, we are given a set of objects with pairwise similarity information. Our aim is to partition these objects into clusters that match this information as closely as possible. The pairwise information is represented as a weighted graph $G$ whose edges are labelled as ``positive/similar'' and ``negative/dissimilar'' by a noisy binary classifier. The goal is to find a clustering $\mathcal{C}$ that is consistent as much as possible with the edge labels. Two objectives for the Correlation Clustering problem have been considered in the literature.

A positive edge is in agreement with $\mathcal{C}$, if its endpoints belong to the same cluster;
and a negative edge is in agreement with $\mathcal{C}$ if its endpoints belong to different clusters. The first objective is to maximize the weight of edges agreeing with $\mathcal{C}$. We call this objective the MaxAgree objective.

Similarly, a positive edge is in disagreement with $\mathcal{C}$, if its endpoints belong to different clusters; and a negative edge is in disagreement with $\mathcal{C}$ if its endpoints belong to the same cluster. The second objective is to minimize the weight of edges disagreeing with $\mathcal{C}$. We call this objective the MinDisagree objective. These two objectives are equivalent at optimality, but, differ from the approximation perspective.

Observe that if a binary classifier is not noisy, i.e., if there exists a clustering which is consistent with all the edge labels, then it is easy to find: simply remove all negative/dissimilar edges and output the connected components of the remaining graph. Thus, the interesting case is when the binary classifier is noisy.

Both MaxAgree and MinDisagree objectives have been extensively studied in the literature since they were introduced by~\citet*{BBC04}. There are currently two standard settings for Correlation Clustering which we will refer to as (1) Correlation Clustering on Complete Graphs and (2) Correlation Clustering with Noisy Partial Information. In the former setting, we assume that graph $G$ is complete and all edge weights are the same, i.e., $G$ is unweighted. In the latter setting, we do not make any assumptions on the graph $G$. Thus, edges can have arbitrary weights and some edges may be missing. These settings are quite different
from the computational perspective.
\paragraph{Correlation Clustering on Complete Graphs:} For MinDisagree objective the first constant-factor approximation algorithm was given by~\citet*{BBC04}. \citet*{CGW03} gave a $4-$approximation algorithm for the problem. \citet*{ACN08} gave two different algorithms with approximation factors of $3$ and $2.5$. Finally, a $2.06-$approximation algorithm was given by~\citet*{CMSY15}. This is currently the best approximation guarantee for MinDisagree objective.

However, MaxAgree objective is easier than MinDisagree objective. There is a trivial $2-$approximation algorithm: if the number of positive edges is greater than the number of negative edges output $G$ as a single cluster, otherwise output each vertex as a singleton cluster. Furthermore, \citet{BBC04} gave a polynomial time approximation scheme (PTAS) for MaxAgree objective, i.e., it can be approximated within any constant in polynomial time. On the other hand,~\citet*{CGW03} proved that MinDisagree objective is APX-hard, i.e., it is NP-hard to approximate it within a certain constant $c>1.$
\paragraph{Correlation Clustering with Noisy Partial Information:} For MinDisagree objective \citet*{CGW03} and \citet*{DEFI06} independently gave an $O(\log n)$ approximation algorithm. Both algorithms use a linear programming relaxation and are heavily inspired by the region growing approach for the Multicut problem by~\citet*{GVY96}. Both \citet{CGW03} and \citet{DEFI06} complement their results with a matching integrality gap of $\Omega(\log n)$, thus implying that one cannot hope for a better approximation guarantee using an LP-based approach. Interestingly,~\citet{CGW03} and \citet{DEFI06} showed that this inspiration by~\citet{GVY96} was not coincidental: they proved that in this setting MinDisagree objective is equivalent to Multicut from the approximation perspective. In particular, they gave an approximation preserving reduction for both directions between two problems. This combined with the hardness result for Multicut by~\citet*{CKKRS06} show that $O(\log n)$ is likely to be the best possible approximation for this problem. More specifically, it is NP-hard to obtain any constant factor approximation algorithm for MinDisagree if the Unique Games Conjecture is true.

For MaxAgree objective \citet*{CGW03} and~\citet*{Swamy04} gave a $0.766-$approximation algorithm based on semidefinite programming (SDP) relaxation. Furthermore, \citet{CGW03} proved that MaxAgree is APX-hard. For other models of Correlation Clustering see~\cite{CDK14},~\cite{AddadLMNP21},~\cite{ACG+15, AW21},~\cite{MSS10},~\cite{LMV+21},~\cite{AALZ12},~\cite{CMSY15},~\cite{MW10},~\cite{MMV15-CC},~\cite{JafarovKMM20},~\cite{jafarov21a}.

\paragraph{}Correlation Clustering is different from other classical clustering problems in that the given data is qualitative rather than quantitative. More specifically, in Correlation Clustering objects are embedded in a graph and a similarity information between two objects is given as an edge label as opposed to the $k$-means or $k$-median clustering problems where objects are embedded in a metric space and a similarity information between two objects is given by a distance function. Furthermore, the number of clusters $k$ is not given as a separate parameter; the number of clusters in an optimal clustering solely depends on the instance. Thus, correlation clustering finds applications in many different problems in machine learning, biology, data mining and other areas of science and engineering (see~\cite{tang2016multi},~\cite{pan2015},~\cite{FS03},~\cite{CR01},~\cite{BSY99}).
\section{Organization}
In this exposition we focus on MinDisagree objective in the Correlation Clustering on Complete Graphs setting which will be simply referred as the Correlation Clustering problem for the rest of the exposition. In Chapter~\ref{Notations} we introduce some notations and definitions which will be useful throughout the exposition. In Chapter~\ref{BBC} we describe the first constant-factor approximation algorithm for the Correlation Clustering problem with an approximation guarantee of $17429$ by~\citet*{BBC04}. In Chapter~\ref{CGW} we analyze a $4-$approximation algorithm based on a natural linear programming (LP) relaxation by~\citet*{CGW03}. In Chapter~\ref{gap_CGW} we show an integrality gap of $2$ for this LP relaxation by~\citet{CGW03}. In Chapter~\ref{ACN} we describe a $3-$approximation algorithm by~\citet*{ACN08}. It is a simple randomized algorithm which is not based on a LP relaxation. Finally, in Chapter~\ref{CMSY} we describe a $2.06-$approximation algorithm by~\citet*{CMSY15}. The algorithm is based on the LP relaxation introduced by~\citet*{CGW03}. It is currently the best approximation algorithm for the Correlation Clustering problem and its approximation factor almost matches with the integrality gap of $2$, given by~\citet{CGW03}.

%% file: Notations.tex
\chapter{Notations and Definitions}\label{Notations}
Let $G=(V,E^+,E^-)$ be a complete unweighted graph where $E^+$ is the set of edges labeled as positive/similar, and $E^-$ is the set of edges labeled as negative/dissimilar. We use $``\PlusSign"$ and $``\MinusSign"$ to denote labels of positive and negative edges, respectively. Let $N^+(u)=\{v\vert\; (u,v)\in E^+\}\cup \{u\}$ and $N^-(u)=\{v\vert\; (u,v)\in E^-\}$ denote the vertices connected to $u$ with positive and negative edges, respectively.

In a clustering $\mathcal{C}$, we call an edge $(u,v)$ a mistake if it is in disagreement with $\mathcal{C}$. When $(u,v)\in E^+$, we call the mistake a positive mistake, otherwise it is called a negative mistake. We denote the total number of mistakes of a clustering $\mathcal{C}$ by $m_{\mathcal{C}}.$ Furthermore, we use $m^+_{\mathcal{C}}$ and $m^-_{\mathcal{C}}$ to denote the number of positive and negative mistakes of $\mathcal{C}$, respectively.

Finally, let $\OPT$ denote the optimal clustering on $G$ and $\LP$ denote the value of an optimal solution of the linear program defined in Chapter~\ref{CGW}.

%% file: Bansal-Blum-Chawla-Algorithm.tex
\chapter{Bansal-Blum-Chawla Algorithm}\label{BBC}
In this chapter we present the $17429-$approximation algorithm for the Correlation Clustering problem by~\citet*{BBC04}. Although the factor is large, it is the first algorithm with any approximation guarantee.

Firstly, we show a natural lower bound for $m_{\OPT}$ which will be used to upper bound the number of mistakes made by the algorithm. We start with defining a ``bad" triangle.
\begin{definition}
Let $T$ be a triangle (a complete graph on $3$ vertices). $T$ is bad if it contains two positive edges and one negative edge.
\end{definition}

Observe that a bad triangle always leads to a mistake regardless of how its endpoints are clustered. Therefore, the number of edge disjoint bad triangles is a lower bound for $\OPT.$

\begin{definition}
Let $\mathcal{T}$ denote the set of bad triangles and assign to each bad triangle $T\in \mathcal{T}$ a nonnegative real number $r_T\in [0,1]$. A set $\{r_T\}_{T\in\mathcal{T}}$ is called a fractional packing of bad triangles if
\begin{align}\label{fractionalPacking}
    \sum\limits_{\substack{T\in \mathcal{T}: \\ e\in T}}r_T\leq 1, \forall e\in E.
\end{align}
\end{definition}
We generalize the above observation as follows:
\begin{lemma}\label{fractionalPacking_Lemma_BBC}
For any fractional packing of bad triangles $\{r_T\}_{T\in \mathcal{T}}$, we have $\sum\limits_{T\in \mathcal{T}} r_T\leq m_{\OPT}.$
\end{lemma}
\begin{proof}
Let $M$ be the set of edges in disagreement with $\OPT.$ Then
\begin{align*}
m_{\OPT}=\sum\limits_{e\in M}1\geq \sum\limits_{e\in M}\sum\limits_{\substack{T\in \mathcal{T}: \\ e\in T}}r_T=\sum\limits_{T\in \mathcal{T}}\vert M\cap T\vert r_T\geq \sum\limits_{T\in \mathcal{T}}r_T.
\end{align*}
where the first inequality follows from~(\ref{fractionalPacking}) and the second inequality follows from $\vert M\cap T \vert\geq 1$ since $\OPT$ always makes at least one mistake on a bad triangle. 
\end{proof}

 Before delving into details we give an outline of the proof, which can be divided into three steps: 
 \begin{itemize}
     \item Show that the number of mistakes of a so-called ``clean" clustering is close to $m_{\OPT}$. 
     \item Prove the existence of a clean clustering by manipulating $\OPT.$
     \item Give an algorithm to find a clean clustering.
 \end{itemize}
We define a ``good" vertex which is central to the definition of a clean clustering. Loosely speaking, a vertex is good if it is assigned to a cluster such that most of its incident edges are not in disagreement with the clustering, and a clustering is clean if all its vertices are good.
\begin{definition}
A vertex $v$ is called $\delta-$good with respect to $C$, where $C\subseteq V$, if it satisfies the following:
\begin{itemize}
    \item[$-$] $\vert N^{+}(v)\cap C \vert\geq (1-\delta)\vert C\vert$
    \item[$-$] $\vert N^{+}(v)\cap V\setminus C\vert \leq \delta \vert C\vert$.
\end{itemize}
If a vertex $v$ is not $\delta-$good with respect to $C$, then it is called $\delta-$bad
with respect to $C$. Finally, a set $C$ is $\delta-$clean if all $v\in C$ are $\delta-$good $C$.
\end{definition}
Intuitively, one should expect the number of mistakes of a clean clustering to be close to the number of mistakes of $\OPT$ since most of the edges are not in disagreement with a clean clustering. We formalize this in the following lemma.
\begin{lemma}\label{clean_lemma_BBC}
Given a clustering of $V$ in which all clusters are $\delta$-clean for
some $\delta\leq \frac{1}{4}$, then there exists a fractional packing $\{r_T\}_{T\in\mathcal{T}}$ such that the number of mistakes made by this clustering is at most $4\sum\limits_{T\in \mathcal{T}}r_T.$
\end{lemma}
\begin{proof}
Firstly, we give a procedure to associate a negative mistake $(u,v)$ with a bad triangle $T_{uvw}$ such that $u,v,w\in C_i$ for some $i\in [k].$ Then we show that the number of bad triangles chosen by this procedure sharing an edge is bounded above by some constant. This gives us a fractional packing of bad triangles with small total value.

Let $\mathcal{C}=\{C_1,\dots, C_k\}$ be a $\delta-$clean clustering. Let $(u,v)\in C_i\times C_i$ be a negative edge in disagreement that has not been considered so far. We want to choose a vertex $w\in C_i$ such that \textbf{(1)} both $(u,w)$ and $(v,w)$ are positive, and \textbf{(2)} it has not been picked by a negative edge in disagreement $(u,v')$ or $(u',v)$ in the previous steps. We associate $(u,v)$ with the bad triangle $T_{uvw}=(u,v,w).$ Observe that the second condition can be thought of as an attempt to enforce edge-disjointness on the triangles chosen by the procedure. We now show that such $w\in C_i$, in other words, such a bad triangle $T_{uvw}$ always exists.

Since $C_i$ is $\delta-$clean, we have $\vert N^{-}(u)\cap C_i\vert \leq \delta \vert C_i \vert$ and $\vert N^{-}(v)\cap C_i\vert \leq \delta \vert C_i \vert$. Thus there exist at least $(1-2\delta)\vert C_i \vert$ many vertices $w\in C_i$ such that both $(u,w),(v,w)\in E^{+}.$ For the same reasons, at most $(2\delta- 2)\vert C_i\vert$ many of these vertices might have been chosen by negative mistakes $(u,v')$ or $(u',v)$ before. Therefore, there are at least $(1-4\delta)\vert C_i \vert + 2$ many vertices available for $(u,v)$. This number is always positive since $\delta\leq \frac{1}{4}$.

Now we bound the number of bad triangles sharing a certain edge. Observe that a positive edge $(u,w)\in C_i\times C_i$ can appear in at most two bad triangles one due to a negative mistake $(u,v)\in C_i\times C_i$ and second possibly due to a negative mistake $(w,x)\in C_i\times C_i$. Furthermore, no edge having end points in different clusters can appear since the procedure chooses only in-cluster bad triangles. Thus, assigning $r_{uvw}=\frac{1}{2}$ for bad triangles $T_{uvw}$ chosen by the procedure and $r_{uvw}=0$ for bad triangles not chosen by the procedure we satisfy~(\ref{fractionalPacking}). Since for each negative mistake a bad triangle is picked, we have a fractional packing of bad triangles such that $\sum\limits r_{uvw}\geq \frac{1}{2}m^{-}_{\mathcal{C}}. $ 

Secondly, we give a similar procedure to associate a positive mistake $(u,v)$ with a bad triangle $T_{uvw}$ such that $u,v,w\in C_i\cup C_j$ for some $i,j\in [k].$ Using similar arguments we show that each edge can appear in at most constant number of bad triangles chosen by the procedure. This gives us a fractional packing of bad triangles with small total value which is different from the one constructed for negative mistakes.

Let $(u,v)\in C_i\times C_j$ for some $i\neq j$, be a positive edge in disagreement that has not been considered so far. Without loss of generality let $\vert C_i \vert\geq \vert C_j \vert$. We want to choose a vertex $w\in C_i$ such that \textbf{(1)} $(u,w)\in E^{+}$ but $(v,w)\in E^{-}$, and \textbf{(2)} it has not been picked by a positive mistake $(u,v')\in C_i\times C_k$ for some $i\neq k$ or by a positive mistake $(u',v)\in C_i\times C_j$ in the previous steps. We associate $(u,v)$ with the bad triangle $T_{uvw}=(u,v,w).$ Observe that the second condition enforces edge-disjointness on the triangles chosen by the procedure. We now show that such $w\in C_i$, in other words, such a bad triangle $T_{uvw}$ always exists.

Since $C_i$ is $\delta-$clean we have $\vert N^{+}(u)\cap C_i\vert \geq (1-\delta)\vert C_i\vert$. Similarly, since $C_j$ is $\delta-$clean we have $\vert N^{+}(v)\cap C_i\vert \leq \vert N^{+}(v)\cap \overline{C_j}\vert \leq \delta\vert C_j\vert\leq \delta \vert C_i \vert$. Thus, there exist at least $(1-2\delta)\vert C_i\vert$ many vertices $w\in C_i$ such that $(u,w)\in E^{+}$ and $(v,w)\in E^{-}.$ For the same reasons, at most $\delta (\vert C_i\vert+\vert C_j\vert)-2\leq 2\delta \vert C_i\vert-2$ many of these vertices might have been chosen by positive mistakes $(u,v')$ or $(u',v)$ before. Therefore, there are at least $(1-4\delta)\vert C_i \vert + 2$ many vertices available for $(u,v)$. This number is always positive since $\delta\leq \frac{1}{4}$.

Now we bound the number of bad triangles sharing a certain edge. Observe that a positive edge $(u,w)\in C_i\times C_i$ can appear in at most two bad triangles chosen by the procedure, one due to a positive mistake $(u,v)\in C_i\times C_j$ and second possibly due to a positive mistake $(w,x)\in C_i\times C_k$ for some $i\neq k$. Furthermore, a negative edge $(w,v)\in C_i\times C_j$ can appear in at most one bad triangle since the procedure picks vertices from $C_i$ and $v\in C_j$. Lastly, a negative edge in disagreement does not appear in any bad triangle. Thus, assigning $r_{uvw}=\frac{1}{2}$ for bad triangles $T_{uvw}$ chosen by the procedure and $r_{uvw}=0$ for bad triangles not chosen by the procedure we satisfy~(\ref{fractionalPacking}). This gives us a fractional packing of bad triangles. Since for each positive mistake a bad triangle is picked we have $\sum\limits r_{uvw}\geq \frac{1}{2}m^{+}_{\mathcal{C}}.$

Now depending on whether there are more negative mistakes or more positive mistakes, we can choose the corresponding fractional packing which gives the following
\begin{align*}
    m_{\mathcal{C}}=m_{\mathcal{C}}^-+m_{\mathcal{C}}^+\leq 2\max \{m_{\mathcal{C}}^-,m_{\mathcal{C}}^+\}\leq 4\sum\limits r_{uvw}.
\end{align*}
This finalizes the proof.
\end{proof}

This combined with Lemma~\ref{fractionalPacking_Lemma_BBC} give the following Corollary.

\begin{corollary}
The number of mistakes of a clustering in which all clusters are $\delta$-clean is at most $4m_{\OPT}.$
\end{corollary}

Next we show the existence of a clean clustering by manipulating an optimal clustering. Loosely speaking, we keep the clusters in an optimal clustering which contain small number of bad vertices. Furthermore, we split clusters having large number of bad vertices into singleton clusters. Then we show that the number of mistakes of the new clustering is close to $m_{\OPT}.$ We give a pseudo-code for this algorithm in Algorithm~\ref{algorithm_clean_up_BBC}.
\begin{lemma}\label{clean_exists_lemma_BBC}
There exists a clustering $\OPT'$ in which each non-singleton cluster is $\delta-$clean, and $m_{\OPT'}\leq (\frac{9}{\delta^2}+1)m_{\OPT}.$
\end{lemma}
\begin{proof}
Let $\OPT=\{C_1,\dots C_k\}$ be an optimal clustering. Let $\OPT'$ be a clustering obtained by applying Algorithm~\ref{algorithm_clean_up_BBC} to $\OPT$.
\begin{algorithm}
   \caption{Procedure $\delta$-Clean-Up:}
   \label{algorithm_clean_up_BBC}
\begin{algorithmic}[1]
   \STATE Set $S=\emptyset.$
   \STATE Let $B_i$ be the set of $\frac{\delta}{3}-bad$ vertices in $C_i$ for all $i\in [k]$.
   \FOR{$i=1,\dots,k$}
   \IF{$\vert  B_i \vert\geq\frac{\delta}{3}\vert C_i\vert$}
   \STATE $S=S\cup C_i$, and set $C'_i=\emptyset.$ We call this ``dissolving" the cluster.
   \ELSE
   \STATE $S=S\cup B_i$ and set $C'_i=C_i\setminus B_i$.
   \ENDIF
   \ENDFOR
   \RETURN the clustering $\OPT'=\{C'_1,\dots, C'_k,\{x\}_{x\in S}\}.$
\end{algorithmic}
\end{algorithm}

We prove that $m_{\OPT}$ and $m_{\OPT'}$ are comparable. Firstly, we show that $C'_i$ is $\delta-$clean for all $i\in [k]$. Clearly, this holds if $C'_i=\emptyset.$ Otherwise for each $v\in C'_i,$ we have:
\begin{align*}
    \vert N^+(v)\cap C'_i \vert\geq (1-\frac{\delta}{3})\vert C_i\vert -\frac{\delta}{3}\vert C_i \vert= (1-\frac{2\delta}{3})\vert C_i \vert >(1-\delta)\vert C'_i \vert.
\end{align*}
The first inequality follows from the fact that $v$ is $\frac{\delta}{3}-$good with respect to $C_i$, therefore $\vert N^{+}(v)\cap C_i\vert\geq (1-\frac{\delta}{3})\vert C_i\vert$ and at most $\vert B_i\vert\leq \frac{\delta}{3}\vert C_i\vert$ of these neighbors can be $\frac{\delta}{3}-$bad with respect to $C_i$. The second inequality follows from $\vert C_i\vert\geq \vert C'_i\vert$.

Similarly, since $v$ is $\frac{\delta}{3}-$good with respect to $C_i$ we have $\vert N^{+}(v)\cap \overline{C_i}\vert\leq \frac{\delta}{3}\vert C_i\vert$. This together with $\vert C_i\setminus C'_i \vert\leq\frac{\delta}{3} \vert C_i\vert$ give 
\begin{align*}
    \vert N^+(v)\cap \overline{C'_i} \vert\leq \frac{\delta}{3}\vert C_i\vert +\frac{\delta}{3}\vert C_i \vert\leq \frac{2\delta}{3} \frac{\vert C_i\vert}{(1-\delta/3)}  <\delta\vert C'_i \vert.
\end{align*}
where the last two inequalities follow from $\vert C'_i \vert\geq (1-\frac{\delta}{3}) \vert C_i\vert$ and $\delta <1.$ Thus $C'_i$ is $\delta-$clean.

We now upper bound the number of mistakes in $\OPT'$. Observe that if $C_i$ is dissolved then the number of mistakes with an endpoint in $C_i$ must be at least $\frac{1}{2}\frac{\delta^2\vert C_i\vert^2}{9}$ since $\vert B_i \vert\geq \frac{\delta}{3}\vert C_i\vert.$ Furthermore, the number of mistakes added due to dissolving $C_i$ is at most $\frac{\vert C_i \vert^2}{2}$.

In case $C_i$ is not dissolved, the number of mistakes with an endpoint in $C_i$ must be at least $ \frac{\delta\vert C_i\vert \vert B_i \vert}{2\cdot 3}.$ However, the number of mistakes added is at most $\vert B_i \vert \vert C_i \vert$. Since $\frac{6}{\delta}<\frac{9}{\delta^2}$ the extra cost of ``cleaning" $\OPT$ is at most $\frac{9}{\delta^2}m_{\OPT}$ and the lemma follows. 
\end{proof}
For the clustering $\OPT'$ given by Algorithm~\ref{algorithm_clean_up_BBC}, we use $C'_i$ to denote the non-singleton clusters and $S$ to denote the set of singleton clusters throughout the rest of the proof. We now give an algorithm to find a clean clustering whose cost is comparable to the cost of $\OPT'.$ Loosely speaking, the algorithm iteratively picks an arbitrary pivot and considers the set of active vertices in the positive neighborhood of the pivot. In the first phase it removes $3\delta-$bad vertices from the set one by one. In the second phase it adds $7\delta-$good vertices to the set all at once. We give a pseudo-code for this algorithm in Algorithm~\ref{algorithm_cautious_BBC}. We assume $\delta=\frac{1}{44}$.
\begin{algorithm}
   \caption{Procedure Cautious}
   \label{algorithm_cautious_BBC}
\begin{algorithmic}[1]
   \STATE Set $V_{active}=V$, $Z=\emptyset$ and $\mathcal{I}=\emptyset.$
   \WHILE{$V_{active}\neq Z$}
   \STATE Pick an arbitrary vertex $v\in V_{active}$.
   \STATE Set $A(v)=N^+(v)\cap V_{active}.$
   \STATE (\textbf{Vertex Removal Step})
   \WHILE{$\exists x\in A(v)$ such that $x$ is $3\delta-$bad with respect to $A(v)$}
   \STATE $A(v)=A(v)\setminus \{x\}$
   \ENDWHILE
   \STATE (\textbf{Vertex Addition Step})
   \STATE Let $Y=\{y\in V_{active}\vert\; y\mbox{ is } 7\delta-$\mbox{ good with respect to }A(v)$\}$.
   \STATE Set $A(v)=A(v)\cup Y$
   \IF {$A(v)=\emptyset$}
   \STATE Set $Z=Z\cup \{v\}$ 
   \ELSE
   \STATE $V_{active}=V_{active}\setminus A(v)$ and set $\mathcal{I}=\mathcal{I}\cup \{v\}.$
   \ENDIF
   \ENDWHILE
   \RETURN the clustering $\mathcal{A}=\bigcup\limits_{v\in \mathcal{I}}\{A(v)\} \cup \bigcup\limits_{v\in Z}\{\{v\}\}.$
\end{algorithmic}
\end{algorithm}

Observe that in the vertex addition step, all vertices are added in one step as opposed to in the vertex removal step. Furthermore, it is possible that $A(v)$ becomes empty at the end of vertex removal step. For instance, if a vertex $v$ is  picked such that $\vert N^+(v)\vert$ is bounded by some constant but for all $w\in N^+(v)\setminus \{v\}$, $\vert N^+{w}\vert=\Theta(n)$ then $N^{+}(v)$ cannot survive at the end of the vertex removal step, thus $v$ is added to $Z.$ It is also possible that although $A(v)=\emptyset$ at the end of the vertex removal step, $Y\neq \emptyset$ so $A(v)$ becomes non-empty at the end of the vertex addition step and $v$ is not added to $Z.$ This might happen if there is a vertex which is incident to only negative edges.

Let $\mathcal{A}$  denote the clustering and $A_i$ denote the clusters output by Algorithm~\ref{algorithm_cautious_BBC}. Let $Z$ be the set of singleton clusters created in the final step. In the next theorem we show that $\mathcal{A}$ is clean and closely related to $\OPT'$.
\begin{theorem}\label{main_thm_BBC}
$\forall j,\;\exists i$ such that $C'_j\subseteq A_i.$ Moreover, each $A_i$ is $11\delta-clean.$
\end{theorem}
In order to prove this theorem, we need the following two lemmas.
\begin{lemma}\label{main_lemma_BBC}
If $v\in C'_i,$ where $C'_i$ is a $\delta-$clean cluster in $\OPT'$, then, any vertex $w\in C'_i$ is $3\delta-$good with respect to $N^+(v).$
\end{lemma}
\begin{proof}
Since $C'_i$ is $\delta-$clean and $v,w\in C'_i$ we have
\begin{align*}
    \vert N^+(v)\cap N^+(w) \vert\geq (1-2\delta) \vert C'_i \vert\geq (1-3\delta)\vert N^+(v)\vert
\end{align*}
where the first inequality follows from $\vert N^+(v)\cap C'_i \vert\geq (1-\delta)\vert C'_i\vert$ and $\vert N^+(w)\cap C'_i \vert\geq (1-\delta)\vert C'_i\vert$. The second inequality follows from $\vert N^+(v)\vert \leq (1+\delta) \vert C'_i\vert.$ Similarly,
\begin{align*}
    \vert N^+(v)\cap \overline{N^+(w)} \vert &= \vert N^+(v)\cap \overline{N^+(w)}\cap C'_i\vert +\vert N^+(v)\cap \overline{N^+(w)}\cap \overline{C'_i} \vert\\
    &\leq  2\delta\vert C'_i\vert\leq \frac{2\delta}{1-\delta}\vert N^+(v) \vert\leq 3\delta \vert N^+(v) \vert 
\end{align*}
where the first inequality follows from $\vert N^+(v)\cap \overline{C'_i} \vert\leq \delta \vert C'_i\vert$ and $\vert \overline{N^+(w)}\cap C'_i\vert\leq \delta \vert C'_i\vert.$ The second inequality follows from $(1-\delta)\vert C'_i\vert\leq \vert N^{+}(v)\vert$. 
Thus, $w$ is $3\delta-$good with respect to $N^+(v).$
\end{proof}
\begin{lemma}
Given an arbitrary set $X$, if $v_1\in C'_i$ and $v_2\in C'_j$, then $v_1$ and $v_2$ cannot both be $3\delta-$good with respect to $X$.
\end{lemma}
\begin{proof}
%Firstly, if $v$ is $3\delta-$good w.r.t. $X$, then $(1-\delta)\vert X\vert\leq N^+(v)\leq (1+\delta) \vert X\vert.$
Suppose that $v_1,v_2$ are both $3\delta-$good with respect to $X.$ Then, $\vert N^+(v_1)\cap X \vert\geq (1-3\delta)\vert X\vert$ and $\vert N^+(v_2)\cap X \vert\geq (1-3\delta)\vert X\vert$. Hence, $\vert N^+(v_1)\cap  N^+(v_2) \cap X \vert\geq (1-6\delta)\vert X\vert$, which implies that
\begin{align}\label{small_equation_1_BBC}
    \vert N^+(v_1)\cap  N^+(v_2) \vert\geq (1-6\delta)\vert X\vert.
\end{align}

Also, since $C'_i,C'_j$ are both $\delta-$clean we have $\vert N^+(v_1)\cap \overline{C'_i}\vert\leq \delta \vert C'_i \vert,$ $\vert N^+(v_2)\cap \overline{C'_j}\vert\leq \delta \vert C'_j \vert$. These combined with $C'_i\cap C'_j=\emptyset$ imply that
\begin{align}\label{small_equation_2_BBC}
    \vert N^+(v_1)\cap N^+(v_2)\vert\leq \delta (\vert C'_i \vert+\vert C'_j \vert).
\end{align}
Observe that
\begin{align*}
 \vert C'_i \vert &=\vert N^+(v_1)\cap C'_i \vert+ \vert \overline{N^+(v_1)}\cap C'_i \vert\leq \vert N^+(v_1)\cap C'_i \vert+\delta \vert C'_i \vert\\
 &=  \vert N^+(v_1)\cap X\cap C'_i \vert + \vert N^+(v_1)\cap \overline{X} \cap C'_i \vert+\delta \vert C'_i \vert \\
 &\leq \vert N^+(v_1)\cap X \cap C'_i \vert+ 3\delta \vert X\vert + \delta \vert C'_i \vert\leq (1+3\delta) \vert X \vert +\delta \vert C'_i \vert
\end{align*}
where the second inequality follows from $v_1$ being $3\delta-$good with respect to $X.$ Therefore we have $\vert C'_i \vert\leq \frac{1+3\delta}{1-\delta}\vert X\vert$. Using similar arguments we get $\vert C'_j \vert\leq \frac{1+3\delta}{1-\delta}\vert X\vert.$ These combined with~(\ref{small_equation_2_BBC}) imply
\begin{align}\label{small_equation_3_BBC}
    \vert N^+(v_1)\cap  N^+(v_2) \vert\leq 2\delta \frac{1+3\delta}{1-\delta}\vert X\vert.
\end{align}
However, since $\delta<\frac{1}{9}$ we have $ 2\delta (1+3\delta)< (1-6\delta)(1-\delta)$. This combined with~(\ref{small_equation_3_BBC}) contradicts with~(\ref{small_equation_1_BBC}). Thus the statement follows.
\end{proof}

\begin{corollary}\label{coroll_BBC}
After the vertex removal step of Algorithm~\ref{algorithm_cautious_BBC}, no two vertices from distinct $C'_i$ and $C'_j$ can be present in $A(v).$
\end{corollary}
Now we prove Theorem~\ref{main_thm_BBC}.
\begin{proof}[Proof of Theorem~\ref{main_thm_BBC}:] For a cluster $A_i$, let $A'_i$ be the set produced after the vertex removal step of Algorithm~\ref{algorithm_cautious_BBC} such that the cluster $A_i$ is obtained by applying the vertex addition step to $A'_i$.

We partition the proof of the theorem into three parts. Firstly, we show that each $A_i$ is either a subset of $S$ or contains exactly one of the clusters $C'_j.$ Secondly, we prove that for all $j$ we can find an $i$ such that $C'_j\subseteq A_i$. Finally, we prove that $A_i$ is $11\delta-$clean for all $i.$

\begin{claim}\label{claim_1_BBC}
For each cluster $A_i$ output by Algorithm~\ref{algorithm_cautious_BBC} either (i) $A_i\subseteq S$ or (ii) there exists unique $j$ such that $C'_j\subseteq A_i$. 
\end{claim}
\begin{proof}
Firstly, we consider the case when $A'_i\subseteq S$. We want to show that during the vertex addition step of Algorithm~\ref{algorithm_cautious_BBC}, no vertex $u\in C'_j$ can enter $A'_i$ for any $j$. For $u$ to enter $A'_i$ one of the necessary conditions is $\vert N^{+}(u)\cap A'_i\vert\geq (1-7\delta)\vert A'_i\vert$. Then, since $\vert N^{+}(u)\cap\overline{C'_i}\vert \leq \delta \vert C'_i\vert$ and $C'_j\cap A'_i=\emptyset$ it follows that
\begin{align}\label{small_equation_4_BBC}
    \delta \vert C'_i\vert\geq (1-7\delta)\vert A'_i\vert
\end{align}

The other necessary condition is $\vert N^{+}(u)\cap \overline{A'_i}\vert\leq 7\delta \vert A'_i \vert.$ Since $\vert N^{+}(u)\cap C'_i\vert\geq (1-\delta)\vert C'_i\vert$ and $C'_j\cap A'_i=\emptyset$ we have
\begin{align}\label{small_equation_5_BBC}
    (1-\delta)\vert C'_j\vert\leq 7\delta\vert A'_i\vert
\end{align}
However since $\delta=\frac{1}{44}$, (\ref{small_equation_4_BBC}) and (\ref{small_equation_5_BBC}) contradict with each other. Thus $A_i\subseteq S.$

Now we consider the second case. In particular, we assume that there exists $u\in C'_j$ for some $j$ such that $u\in A'_i.$ We want to show that (1) $C'_j\subseteq A_i$ and (2) $A_i\cap C'_k=\emptyset$ for all $k\neq j$. First, we show that (2) holds. Observe that no vertex $v\in C'_k$ for $k\neq j$ can appear in $A'_i$ due to Corollary~\ref{coroll_BBC}. For $u\in C'_k$ to enter $A'_i$ during the vertex addition step of Algorithm~\ref{algorithm_cautious_BBC} one of the necessary conditions is $\vert N^{+}(v)\cap A'_i\vert\geq (1-7\delta)\vert A'_i\vert.$ Furthermore, $\vert N^{+}(u)\cap A'_i\vert \geq (1-3\delta)\vert A'_i\vert$ which implies that
\begin{align}\label{small_equation_6_BBC}
    \vert N^{+}(u)\cap N^{+}(v)\vert\geq (1-10\delta)\vert A'_i\vert
\end{align}
Similarly, $\vert N^{+}(v)\cap \overline{C'_k}\vert\leq \delta\vert C'_k\vert$ and $\vert N^{+}(u)\cap \overline{C'_j}\vert\leq \delta\vert C'_j\vert$ hold. Since $C'_j\cap C'_k=\emptyset$ it follows that 
\begin{align}\label{small_equation_7_BBC}
    \vert N^{+}(u)\cap N^{+}(v)\vert\leq \delta(\vert C'_j\vert +\vert C'_k\vert)
\end{align}
Observe that
\begin{align}\label{small_equation_9_BBC}
 (1-\delta)\vert C'_j\vert \leq \vert N^{+}(u)\vert\leq (1+3\delta)\vert A'_i\vert  
\end{align}
and
\begin{align}\label{small_equation_10_BBC}
    (1-\delta)\vert C'_k\vert \leq\vert N^{+}(v)\vert\leq (1+7\delta)\vert A'_i\vert
\end{align}
Then~(\ref{small_equation_9_BBC}), (\ref{small_equation_10_BBC}) together with (\ref{small_equation_7_BBC}) imply that
\begin{align}\label{small_equation_8_BBC}
   \vert N^{+}(u)\cap N^{+}(v)\vert\leq \frac{\delta}{1-\delta}(\vert N^{+}(u)\vert+\vert N^{+}(v)\vert)\leq \frac{\delta}{1-\delta}(2+10\delta)\vert A'_i\vert
\end{align}
Since $\delta=\frac{1}{44}$, (\ref{small_equation_8_BBC}) contradicts with (\ref{small_equation_6_BBC}). Thus we show that $A_i\cap C'_k=\emptyset$ for $k\neq j.$

Now we show that condition (1) holds. That is, we want to prove that $C'_j\subseteq A_i.$ Observe that~(\ref{small_equation_9_BBC}) implies $\vert C'_j\vert< 2\vert A'_i\vert$. Then we have 
\begin{align}\label{small_equation_11_BBC}
    \vert A'_i\cap C'_j\vert\geq \vert A'_i\cap N^+(u)\vert-\vert N^+(u)\cap \overline{C'_j}\vert\geq \vert A'_i\cap N^+(u)\vert-\delta \vert C'_j\vert\geq (1-5\delta)\vert A'_i\vert. 
\end{align}

We now show that all remaining vertices from $C'_j$ enter $A_i$ during the vertex addition step of Algorithm~\ref{algorithm_cautious_BBC}. Consider a vertex $w\in C'_j$ such that $w\not\in A'_i$. Since $C'_j$ is $\delta-$clean we have $\vert\overline{N^+(w)}\cap C'_j\vert\leq \delta \vert C'_j\vert$. Furthermore, (\ref{small_equation_11_BBC}) implies $\vert A'_i\cap \overline{C'_j}\vert\leq 5\delta \vert A'_i\vert$. Then it follows that
\begin{align*}
    \vert A'_i\cap \overline{N^+(w)}\vert\leq 5\delta \vert A'_i\vert + \delta \vert C'_j\vert\leq 7\delta \vert A'_i\vert.
\end{align*}
Similarly we have $\vert \overline{C'_j}\cap N^{+}(w)\vert\leq \delta \vert C'_j\vert<2\delta \vert A'_i\vert$ and $\vert \overline{A'_i}\cap C'_j\vert \leq 5\delta \vert A'_i\vert$ together which imply
\begin{align}
    \vert \overline{A'_i}\cap N^+(w)\vert\leq 7\delta \vert A'_i\vert.
\end{align}

Thus $w$ is $7\delta-$good with respect to $A'_i,$ and will be added in the vertex addition step of Algorithm~\ref{algorithm_cautious_BBC}. The claim follows.
\end{proof}

Now we show that for all $j$ there exists an $i$ such that $C'_j\subseteq A_i$.

\begin{claim}\label{claim_2_BBC}
$\forall j$, $\exists i$ such that $C'_j\subseteq A_i.$
\end{claim}
\begin{proof}
Let $v\in C'_j$ be a vertex picked by Algorithm~\ref{algorithm_cautious_BBC}. We show that during the vertex removal step, no vertex $w\in N^+(v) \cap C'_j$ is removed from $A(v)$. The proof follows by an induction on the number of vertices removed so far ($r$) in the vertex removal step by Algorithm~\ref{algorithm_cautious_BBC}. The base case ($r = 0$) follows from Lemma~\ref{main_lemma_BBC} since every vertex in $C'_j$ is $3\delta-$good with respect to $N^+(v)$. For the induction step assume $N^+(v) \cap C'_j\subseteq A(v)$ thus far. We show that every vertex $w\in C'_j$ is still $3\delta-$good with respect to the intermediate $A(v)$. Since $C'_j$ is $\delta-$clean, we have $\vert N^{+}(w)\cap C'_j\vert \geq (1-\delta)\vert C'_j\vert$ and $\vert N^{+}(v)\cap C'_j\vert \geq (1-\delta)\vert C'_j\vert$ which imply $\vert N^{+}(w)\cap N^{+}(v)\cap C'_j\vert \geq (1-2\delta)\vert C'_j\vert$. Then
\begin{align*}
    \vert N^{+}(w)\cap A(v)\vert\geq \vert  N^{+}(w)\cap N^{+}(v)\cap C'_j \vert \geq (1-2\delta)\vert C'_j\vert \geq \frac{1-2\delta}{1+\delta}\vert A(v)\vert > (1-3\delta) \vert A(v)\vert 
\end{align*}
where the third inequality follows from $\vert A(v)\vert\leq \vert N^+(v)\vert\leq (1+\delta)\vert C'_j\vert.$ For the same reasons we have $\vert\overline{A(v)}\cap C'_j\vert\leq \delta\vert C'_j\vert$, $\vert N^{+}(w)\cap \overline{C'_j}\vert \leq \delta\vert C'_j\vert$ and $\vert C'_j\vert \leq \frac{1}{1-\delta}\vert A(v)\vert.$ Then it follows that
\begin{align*}
   \vert N^{+}(w)\cap \overline{A(v)}\vert&=\vert N^{+}(w)\cap \overline{A(v)}\cap \overline{C'_j}\vert+\vert N^{+}(w)\cap \overline{A(v)}\cap C'_j\vert\\
   &\leq 2\delta \vert C'_j\vert\leq \frac{2\delta}{1-\delta} \vert A(v)\vert<3\delta \vert A(v)\vert.
\end{align*}

Thus $w\in C'_j$ is $3\delta-$ good with respect to $A(v)$ is not removed which ends the proof of induction. We proved that $A'_i$ contains at least $(1-\delta)\vert C'_j$ vertices of $C'_j$ at the end of the vertex removal step, and hence by Claim~\ref{coroll_BBC}, $C'_j\subseteq A_i$ after the vertex addition step of Algorithm~\ref{algorithm_cautious_BBC}.
\end{proof}

Finally we show that every non-singleton cluster is $11\delta-$clean.

\begin{claim}\label{claim_3_BBC}
Each non-singleton cluster $A_i$ is $11\delta-$clean.
\end{claim}

\begin{proof}
At the end of vertex removal step, $\forall x\in A'_i$, $x$ is $3\delta-$good with respect to $A'_i$ which implies $\vert N^+(x)\cap \overline{A'_i}\vert\leq 3\delta \vert A'_i\vert.$  So the total number of positive edges leaving $A'_i$ is at most $3\delta \vert A'_i\vert^2.$ Since, in the vertex addition step, we add vertices that are $7\delta-$good with respect to $A'_i$, the number of newly added vertices can be at most $3\delta \vert A'_i\vert^2/(1-7\delta)\vert A'_i\vert< 4\delta \vert A_i \vert$. Thus $\vert A_i\vert \leq (1+4\delta)\vert A'_i\vert.$ Since all vertices $v\in A_i$ are at least $7\delta-$good with respect to $A'_i$ it follows that
\begin{align*}
    \vert N^+(v)\cap A_i\vert\geq (1-7\delta)\vert A'_i\vert\geq \frac{1-7\delta}{1+4\delta}\vert A_i\vert \geq (1-11\delta)\vert A_i\vert.
\end{align*}

Similarly, we have $\vert N^+(v)\cap \overline{A_i}\vert\leq 7\delta\vert A'_i\vert\leq 11\delta \vert A_i\vert.$ And the statement follows.
\end{proof}
Claim~\ref{claim_2_BBC} and Claim~\ref{claim_3_BBC} finalize the proof of Theorem~\ref{main_thm_BBC}.
\end{proof}

Now we bound the number of mistakes of $\mathcal{A}$ output by Algorithm~\ref{algorithm_cautious_BBC} in terms of $\OPT$ and $\OPT'$. Call edges in disagreement that have both end points in some clusters $A_i$ and $A_j$ as internal mistakes and those that have an end point in $Z$ as external mistakes. Similarly in $\OPT'$, call edges in disagreement that have both end points in some clusters $C'_i$ and $C'_j$ as internal mistakes and those having one end point in $S$ as external mistakes. We bound the mistakes in two steps: The following lemma bounds external mistakes.

\begin{lemma}\label{external_mistakes_BBC}
The total number of external mistakes of Algorithm~\ref{algorithm_cautious_BBC} is at most the number of external mistakes of $\OPT'$.
\end{lemma}
\begin{proof}
From Theorem~\ref{main_thm_BBC}, it follows that $Z\cap C'_i=\emptyset$ for all $i$. Thus, $Z\subseteq S$. Any external mistake made by Algorithm~\ref{algorithm_cautious_BBC} corresponds to a positive edge in disagreement which is adjacent to some singleton cluster in $Z$. These edges are also in disagreement with $\OPT'$ since they are incident on singleton clusters in $S$. Hence the lemma follows.
\end{proof}

Now consider the internal mistakes of $\mathcal{A}$. Observe that it is sufficient to consider the graph induced by $V'= \bigcup\limits_i A_i$. Furthermore, the cost of the optimal clustering on the graph induced by $V'$ is at most $m_{\OPT}$. Since $11\delta\leq \frac{1}{4}$ we can apply Lemma~\ref{clean_lemma_BBC} to bound the number of internal mistakes.
\begin{lemma}\label{internal_mistakes_BBC}
The total number of internal mistakes of Algorithm~\ref{algorithm_cautious_BBC} is at most  $4m_{\OPT}$.
\end{lemma}

Summing up results from the Lemma~\ref{external_mistakes_BBC} and~\ref{internal_mistakes_BBC}, and using Lemma~\ref{clean_exists_lemma_BBC}, we get the main theorem of this chapter.
\begin{theorem}
Algorithm~\ref{algorithm_cautious_BBC} gives a $(\frac{9}{\delta^2}+5)-$approximation for $\delta=\frac{1}{44}$.
\end{theorem}

%% file: Charikar-Guruswami-Wirth-Algorithm.tex
\chapter{Charikar-Guruswami-Wirth Algorithm}\label{CGW}
In this chapter, we describe the $4-$approximation algorithm for the Correlation Clustering problem by~\citet*{CGW03}. It is a considerable improvement compared to $17429$, the approximation guarantee of the algorithm described in Chapter~\ref{BBC}. In addition to the improvement in the approximation factor, \citet{CGW03} introduce a linear programming (LP) relaxation which has now become the dominant approach to attack this problem.

We start with describing the LP relaxation. We first give an integer programming formulation
of the Correlation Clustering problem. For every pair of vertices $u$ and $v$, the integer program (IP)
has a variable $x_{uv}\in \{0,1\}$, which indicates whether $u$ and $v$ belong to the same cluster:

\[
x_{uv} = \begin{dcases*}
        0  & if $u$ and $v$ belong to the same cluster;\\
        1 & otherwise.
        \end{dcases*}
\]
\iffalse
\begin{itemize}
  \item $x_{uv}=0$, if $u$ and $v$ belong to the same cluster; and
  \item $x_{uv}=1$, otherwise.
\end{itemize}
\fi
We require that $x_{uv}=x_{vu}$, $x_{uu}=0$
and all $x_{uv}$ satisfy the triangle inequality. That is, $x_{uv} + x_{vw}\geq x_{uw}$.

Every feasible IP solution $x$ defines a clustering $\mathcal{C}=(C_1,\dots,C_T)$ in which
two vertices $u$ and $v$ belong to the same cluster if and only if $x_{uv} = 0$.
A positive edge $(u,v)$ is in disagreement with this clustering if and only if $x_{uv} = 1$;
a negative edge $(u,v)$ is in disagreement with this clustering if and only if $x_{uv} = 0$.
Thus, the cost of the clustering is given by the following linear function:
$$\sum_{(u,v)\in E^+} x_{uv} + \sum_{(u,v)\in E^-} (1 - x_{uv}).$$

We now replace all integrality constraints $x_{uv}\in \{0,1\}$ in the integer program
with linear constraints $x_{uv}\in [0,1]$. The obtained linear program is given in Figure~\ref{fig:LP}.
We refer to each variable $x_{uv}$ as the length of the edge $(u,v)$.
\begin{figure}
\openLP

$$\min \sum_{(u,v)\in E^+} x_{uv} + \sum_{(u,v)\in E^-} (1 - x_{uv}).$$

\noindent\textbf{subject to}
\begin{align*}
x_{uw}&\leq x_{uv}+x_{vw}&\text{for all } u,v,w\in V\\
x_{uv}&=x_{vu}&\text{for all } u,v\in V\\
x_{uu}&=0&\text{for all } u\in V\\
x_{uv}&\in [0,1]&\text{for all } u,v\in V
\end{align*}
\closeLP
\caption{Standard LP relaxation}\label{fig:LP}
\end{figure}

Now we describe the algorithm. It takes an optimal solution of this LP as an input and rounds it to an integral solution. Loosely speaking, it iteratively picks an arbitrary pivot and clusters together the vertices adjacent to the pivot with edges of LP length at most $\frac{1}{2}$. In case the average LP weight of this set is at most $\frac{1}{4}$, the algorithm outputs it as a cluster, otherwise discards it and outputs the pivot as a singleton cluster. We give a pseudo-code for this algorithm in Algorithm~\ref{algo_CGW}.

We start with an informal explanation to give an intuition about the algorithm. Let $u=p_{t}$ be the pivot picked at step $t$ by Algorithm~\ref{algo_CGW}. Let $C_t$ denote the cluster constructed at this step. Consider a negative mistake $(i,j)\in E^-$ such that $i,j\in C_t.$ If $x_{ui}$ and $x_{uj}$ are small then LP weight of $(i,j)$ is high since $1-x_{ij}\geq 1-(x_{ui}+x_{uj}).$ Thus, we can charge the mistake $(i,j)$ to its high LP weight. However if $x_{ui}$ and $x_{uj}$ are close to $\frac{1}{2}$ rather than $0$ then $1-(x_{ui}+x_{uj})$ is small. In this case, $(i,j)$ doesn't have a strong guarantee to ``pay" for itself.

Similarly, consider a positive mistake $(i,j)\in E^+$ such that $i\in C_t$ and $j\not\in C_t.$ If $x_{ui}$ is small then $x_{ij}\geq x_{uj}-x_{ui}\geq \frac{1}{2}-x_{ui}\approx \frac{1}{2}$. Thus, we can charge the mistake $(i,j)$ to its high LP weight. However, in case $x_{ui}$ is close to $\frac{1}{2}$, it is possible that $x_{ij}$ is small. Again $(i,j)$ doesn't have a strong guarantee to ``pay" for itself. In other words, having many vertices in $C_t$ far away from $u$ is problematic. To eliminate this possibility we check whether the ``mass" of the cluster is distributed on its boundary or not in Step 6 of Algorithm~\ref{algo_CGW}.
\begin{algorithm}
   \caption{Rounding Algorithm}
   \label{algo_CGW}
\begin{algorithmic}[1]
   \STATE Let $V_1=V$ be the set of active vertices and $t=1$.
   \WHILE{$V_t\neq \emptyset$}
   \STATE Pick an arbitrary pivot $p_t\in V_t$.
   \STATE Let $S=\{v\in V_t\vert\; x_{p_{t}v}\leq \frac{1}{2}\}\setminus\{p_t\}$.
   \STATE Create a new cluster $C_t$ and add $p_t$ to $C_t$.
   \IF{$\frac{1}{\vert S\vert}\sum\limits_{v\in S}x_{p_{t}v}\leq \frac{1}{4}$}
   \STATE $C_t=C_t\cup S.$
   \ENDIF
   \STATE Let $V_{t+1}=V_t\setminus C_t$ and $t=t+1.$
   \ENDWHILE
   \RETURN $\mathcal{C}=\{C_1,\dots, C_{t-1}\}$.
\end{algorithmic}
\end{algorithm}

We formalize this discussion in the main theorem of this chapter.
\begin{theorem}\label{main_thm_CGW}
Algorithm~\ref{algo_CGW} gives a $4-$approximation.
\end{theorem}
\begin{proof}
We prove the result by showing that the cost of Algorithm~\ref{algo_CGW} is at most $4\LP.$ Observe that at the end of step $t$ of Algorithm~\ref{algo_CGW} we remove the edges having an endpoint in the cluster constructed at step $t$ and never consider them again. Thus, the main strategy is bounding the number of mistakes made at step $t$ by the LP weight of the edges removed at step $t$. Let $u=p_t$ be a pivot picked at step $t$. The analysis is split into two cases: (i) $C_t$ is a singleton or (ii) $C_t$ contains multiple vertices.

Consider the case when $C_t=\{u\}$. The edges incident to $u$ are removed at step $t$ and only those with positive labels are in disagreement with $C_t.$ Furthermore, we have $
\sum\limits_{v\in S}x_{uv}> \frac{\vert S\vert}{4}.$ For a negative edge $ (u,v) $ such that $v\in S$, its LP weight is $1-x_{uv}\geq x_{uv}$ since $x_{uv}\leq \frac{1}{2}.$ Then,
\begin{align}\label{singleton_CGW}
    \sum\limits_{\substack{v\in S:\\ (u,v)\in E^+}}x_{uv}+\sum\limits_{\substack{v\in S:\\ (u,v)\in E^-}}(1-x_{uv})\geq \frac{\vert S\vert}{4}.
\end{align}

Observe that the number of positive edges $(u,v)$ such that $v\in S$ can be at most $\vert S\vert.$ Thus~(\ref{singleton_CGW}) implies that the number of such edges in disagreement is at most four times the LP weight of edges having an endpoint in $S.$

For a positive edge $(u,v)$ such that $v\not\in S$ its LP weight $x_{uv}\geq\frac{1}{2}.$ Thus, the number of such edges in disagreement is at most twice the LP weight of the edges having an endpoint in $V_t\setminus S.$ This finalizes the proof for $C_t=\{ u\}$ case.

In case $C_t=\{ u\} \cup S$ we analyze negative and positive mistakes separately. We start with negative mistakes.

Consider a negative mistake $(i,j)$ such that $i,j\in C_t.$ Without loss of generality $x_{ui}\leq x_{uj}$, i.e., $i$ is closer to $u$ than $j$. If $x_{uj}\leq \frac{3}{8}$ then $x_{ij}\leq x_{uj}+x_{ui}\leq \frac{3}{8}+\frac{3}{8}=\frac{3}{4}$ and LP weight of $(i,j)$ is at least $\frac{1}{4}$. Now consider the case where $x_{uj}\in(\frac{3}{8},\frac{1}{2}].$ Observe that this is the problematic case. In particular, if $i$ is also close to boundary, i.e., $x_{ui}\approx \frac{1}{2}$ then it is possible that $x_{ij}\approx 1$, thus LP weight of $(i,j)$ is too small to compensate itself.

In this case we charge all the negative edges in disagreement $(i,j)$ such that $x_{ui}\leq x_{uj}$ to its distant endpoint $j$. Furthermore, we associate total LP weight of all the edges $(k,j)$ (positive and negative) such that $x_{uk}\leq x_{uj}$ to its distant endpoint $j$. Observe that such $(k,j)$ is removed at step $t$ since $k\in C_t.$ Due to triangle inequality, this quantity is at least
\begin{align}\label{negative_CGW}
\sum\limits_{\substack{(i,j)\in E^{+}:\\ x_{ui}\leq x_{uj}}}(x_{uj}-x_{ui})+\sum\limits_{\substack{(i,j)\in E^{-}:\\ x_{ui}\leq x_{uj}}}(1-x_{uj}-x_{ui})=p_jx_{uj}+n_j(1-x_{uj})-\sum\limits_{x_{ui}\leq x_{uj}}x_{ui}    
\end{align}
where $p_j$ denotes the number of positive edges $(u,i)$ such that $x_{ui}\leq x_{uj}$. Similarly, $n_j$ denotes the number of negative edges $(u,i)$ such that $x_{ui}\leq x_{uj}$. Observe that
\begin{align}\label{negative_2_CGW}
\frac{\vert S\vert}{4}\geq \sum\limits_{v\in S}x_{uv} = \sum\limits_{x_{ui}\leq x_{uj}}x_{ui} + \sum\limits_{\substack{i\in S:\\ x_{ui}> x_{uj}}}x_{ui}> \sum\limits_{x_{ui}\leq x_{uj}}x_{ui} + \frac{3}{8} (\vert S\vert-p_j-n_j)    
\end{align}
where the first inequality follows from $S\subseteq C_t$ and the second inequality follows from $x_{uj}\geq \frac{3}{8}.$ Since $\vert S\vert \geq p_j +n_j$ (\ref{negative_2_CGW}) implies $\frac{1}{4}(p_j+n_j)\geq \sum\limits_{x_{ui}\leq x_{uj}}x_{ui}. $ This together with~(\ref{negative_CGW}) imply that total LP weight associated with $j$ is at least
\begin{align}\label{negative_3_CGW}
    p_jx_{uj}+n_j(1-x_{uj})-\frac{1}{4}(p_j+n_j)
\end{align}

Observe that the number of negative mistakes charged to $j$ is simply $n_j.$ Furthermore,~(\ref{negative_3_CGW}) is a linear function of $x_{uj}\in (\frac{3}{8}, \frac{1}{2}]$, thus obtaining its maximum at its boundary. It is equal to $\frac{p_j}{8}+\frac{3n_j}{8}$ at $x_{uj}=\frac{3}{8}$ and is equal to $\frac{p_j}{4}+\frac{n_j}{4}$ at $x_{uj}=\frac{1}{2}.$ Since $\min\{\frac{p_j}{8}+\frac{3n_j}{8},\frac{p_j}{4}+\frac{n_j}{4}\}\geq \frac{n_j}{4}$ and $j$ is chosen arbitrarily, we have that the number of negative edges in disagreement with $C_t$ is at most four times the LP weight of the edges having both endpoints in $C_t.$

Now we consider the positive mistakes. We follow a similar argument used for the previous case. Let $(i,j)$ be a positive edge in disagreement. Without loss of generality $x_{ui}\leq x_{uj}$. If $x_{uj}\geq \frac{3}{4}$ then LP weight of $(i,j)$ is at least $ x_{ij}\geq x_{uj}-x_{ui}\geq \frac{1}{4}$. Now consider the case $x_{uj}\in(\frac{1}{2},\frac{3}{4})$. Observe that this is the problematic case. In particular, if $i$ is also close to boundary, i.e., $x_{ui}\approx \frac{1}{2}$ then it is possible that $x_{ij}\approx 0$, thus LP weight $(i,j)$ is too small to compensate for itself.

In this case we charge all the positive edges in disagreement $(i,j)$ such that $x_{ui}\leq x_{uj}$ to its distant endpoint $j$. Furthermore, we associate total LP weight of all the edges $(k,j)$ (positive and negative) such that $k\in C_t $ and $x_{uk}\leq x_{uj}$ to its distant endpoint $j$. Observe that these edges are removed at step $t$. Due to triangle inequality, this quantity is at least
\begin{align}\label{positive_CGW}
\sum\limits_{\substack{(i,j)\in E^{+}:\\ i \in C_t, \\ x_{ui}\leq x_{uj}}}(x_{uj}-x_{ui})+\sum\limits_{\substack{(i,j)\in E^{-}:\\ i \in C_t, \\ x_{ui}\leq x_{uj}}}(1-x_{uj}-x_{ui})=p_jx_{uj}+n_j(1-x_{uj})-\sum\limits_{i\in C_t}x_{ui}    
\end{align}
where $p_j$ denotes the number of violated positive edges with a distant endpoint $j$. Similarly $n_j$ denotes the number of negative edges in $C_t\times \{j\}$. Since we have $\sum\limits_{i\in C_t}x_{ui}\leq \frac{1}{4}\vert S\vert$ and $p_j+n_j=\vert S\vert$, (\ref{positive_CGW}) is at least
\begin{align}\label{positive_2_CGW}
   p_jx_{uj}+n_j(1-x_{uj})-\frac{1}{4}(p_j+n_j)
\end{align}
which is a linear function in $x_{uj}\in (\frac{1}{2},\frac{3}{4})$. Using the  same arguments from previous case it follows that~(\ref{positive_2_CGW}) ranges between $\frac{1}{4}(p_j+n_j)$ and $\frac{1}{2}p_j$. Since $\min\{\frac{1}{4}(p_j+n_j),\frac{1}{2}p_j\}\geq \frac{1}{4}p_j$ and $j$ is chosen arbitrarily, we have that the number of positive edges in disagreement with $C_t$ is at most four times the LP weight of the edges having exactly one endpoint in $C_t$. The statement follows.
\end{proof}

\section{Integrality Gap}\label{gap_CGW}
In this section we show that the integrality gap for the standard LP relaxation is at least $2$. The result is due to~\citet*{CGW03}. In particular, we describe an instance of Correlation Clustering problem with integrality ratio at least $2$. Let $G_{CC}=(V,E^+,E^-)$ be a complete graph with a special vertex $c\in V$ and $\vert V\vert =n+1$. All the edges of $G_{CC}$ are negative except the ones incident to $c$. That is, $E^+=\{(c,v)\vert v\in V\mbox{ and }v\neq c\}$ and $E^-=V\times V\setminus E^+.$
\begin{theorem}
The integrality ratio of the Correlation Clustering instance $G_{CC}=(V,E^+,E^-)$ described above is at least $2$.
\end{theorem}
\begin{proof}
First we show an upper bound on $\LP$ for $G_{CC}$. In particular, we construct a feasible solution with value at most $\frac{n}{2}$. Consider the following solution
\begin{align*}
x_{e}=
\begin{cases} \frac{1}{2}, & \mbox{if }e\in E^+ \\ 1, & \mbox{if }e\in E^-
\end{cases}
\end{align*}

It is easy to verify that $\{x_e\}_{e\in E^+\cup E^-}$ is a feasible solution to the standard LP. Furthermore, its value is $\sum\limits_{e\in E^+}x_e+\sum\limits_{e\in E^-}(1-x_e)=\frac{1}{2}\vert E^+\vert=\frac{n}{2}$. Thus $\LP\leq \frac{n}{2}$. Now we give a lower bound on the number of mistakes of an optimal clustering $\OPT$. Let $C$ be a cluster in $\OPT$ containing $c$, and let $r$ denote the number of remaining vertices in $C$, i.e., $r=\vert C \setminus\{c\}\vert$. The number of edges in disagreement with $C$ is $n-r+\frac{r(r-1)}{2}=n-1+\frac{(r-1)(r-2)}{2}\geq n-1$ for $r\geq 0$. Thus $m_{\OPT}\geq n-1$. It follows that integrality ratio of $G_{CC}$ is at least $\frac{2(n-1)}{n}=2-\frac{2}{n}$ which has limit $2$ as $n$ increases.
\end{proof} 

%% file: Ailon-Charikar-Newman-Algorithm.tex
\chapter{Ailon-Charikar-Newman Algorithm}\label{ACN}
In this chapter, we present the $3-$approximation algorithm for the Correlation Clustering problem by~\citet*{ACN08}. It is a simple randomized algorithm which is not based on LP relaxation. It is worth noting that~\citet{ACN08} also give an LP-based $2.5-$approximation algorithm for the Correlation Clustering problem. We discuss this algorithm as a special case of the algorithm introduced in Chapter~\ref{CMSY}.

Now we describe the algorithm. Loosely speaking, it iteratively picks a random pivot and clusters together the vertices adjacent to the pivot with positive edges. We give a pseudo-code for this algorithm in Algorithm~\ref{algo_ACN}.
\begin{algorithm}
   \caption{Algorithm KwikCluster}
   \label{algo_ACN}
\begin{algorithmic}[1]
   \STATE Let $V_1=V$ be the set of active vertices and $t=1$
   \WHILE{$V_t\neq\emptyset$}
   \STATE Pick a pivot $p_t\in V_t$ uniformly at random.
   \STATE Create a new cluster $C_t$ and add $p_t$ to $C_t$.
   \FORALL{$v\in V_t\mbox{ such that }v\neq p_t$}
   \IF{$(p_t,v)\in E^+$}
   \STATE $C_t=C_t\cup \{ v\}.$
   \ENDIF
   \ENDFOR
   \STATE Let $V_{t+1}=V_t\setminus C_t$ and $t=t+1$.
   \ENDWHILE
   \RETURN $\{C_1,\dots,C_{t-1}\}$.
\end{algorithmic}
\end{algorithm}

As in Chapter~\ref{BBC} we use bad triangles to upper bound the cost of the algorithm.
\begin{lemma}\label{main_lemma_ACN}
There exists a fractional packing of bad triangles, \{$r_T\}_{T\in\mathcal{T}}$, such that the  expected cost of Algorithm~\ref{algo_ACN} is equal to $3\sum\limits_{T\in\mathcal{T}}r_T$.
\end{lemma}
\begin{proof}
Let $(u,v)$ be a positive edge in disagreement. Then there must exist $t\in\mathbb{N},\;w\in V$ such that $w=p_t$ is chosen as a pivot and $ u,v,w \in V_t $ at step $t$ and $\vert C_t\cap \{u,v\}\vert =1$. Observe that for $\vert C_t\cap \{u,v\}\vert =1$ exactly one of $(u,w)$ and $(v,w)$ must be positive, i.e., $T=(u,v,w)$ must be a bad triangle.

Similarly, let $(u,v)$ be a negative edge in disagreement. Then there must exist $t\in\mathbb{N},\;w\in V$ such that $w=p_t$ is chosen as a pivot and $ u,v,w \in V_t $ at step $t$ and $\vert C_t\cap \{u,v\}\vert =2$. Observe that for $\vert C_t\cap \{u,v\}\vert =2$ both $(u,w)$ and $(v,w)$ must be positive, i.e., $T=(u,v,w)$ must be a bad triangle.

In other words, an edge $(u,v)$ is a mistake if and only if there exists a bad triangle $T=(u,v,w)$ such that one of its vertices is chosen as a pivot when all its vertices belong to $V_t$ for some $t$. Thus we can charge each mistake to some bad triangle.

For a bad triangle $T=(u,v,w)$, let $A_T$ denote the event that all three $u,v,w$ are in $V_t$ when one among them is chosen as a pivot for some step $t$. Let $p_T$ denote the probability of $A_T$. Let $B_e$ denote the event that an edge $e$ is in disagreement. Conditioned on the event $A_T$, each one of the vertices of $T$ is chosen as the pivot with probability $\frac{1}{3}$ since a pivot is chosen uniformly at random. Similarly, conditioned on the event $A_T$, an edge $e\in T$ is a mistake due to $T$ with probability $\frac{1}{3}$ since it occurs only when a vertex other than its endpoints is chosen as a pivot. Thus, for an edge $e\in T$
\begin{align*}
    \Pr[B_e\cap A_T]=\Pr[B_e\vert\;A_T]\Pr[A_T]=\frac{1}{3}p_T.
\end{align*}

Observe that for two different bad triangles $T,T'\in\mathcal{T}$ sharing an edge $e$, the events $B_e\cap A_T$ and $B_e\cap A_{T'}$ are disjoint since an edge $e$ can be charged to only one bad triangle containing $e$. Therefore, for all $e\in E$,
\begin{align*}
    \sum\limits_{\substack{T\in\mathcal{T}:\\e\in T}}\frac{1}{3}p_T\leq 1.
\end{align*}

That is, $\{\frac{1}{3}p_T\}_{T\in \mathcal{T}}$ is a fractional packing of bad triangles. Observe that a bad triangle $T$ is charged by a mistake only when $A_T$ occurs, and it can be charged at most once. Since each mistake is charged to some bad triangle the expected cost of the algorithm is equal to $\sum\limits_{T\in\mathcal{T}}p_T.$ More formally, let $\ALG$ denote the cost of Algorithm~\ref{algo_ACN}, then
\begin{align*}
    \mathbb{E}[\ALG]=\sum\limits_{e\in E}\Pr[B_e]=\sum\limits_{e\in E}\sum\limits_{\substack{t\in\mathcal{T}:\\ e\in T}}\Pr[B_e\cap A_T]=\sum\limits_{e\in E}\sum\limits_{\substack{t\in\mathcal{T}:\\ e\in T}}\frac{1}{3}p_T=3\sum\limits_{T\in \mathcal{T}}\frac{1}{3}p_T=\sum\limits_{T\in \mathcal{T}}p_T.
\end{align*}

The statement follows.
\end{proof}

Lemma~\ref{main_lemma_ACN} combined with Lemma~\ref{fractionalPacking_Lemma_BBC} give the main theorem of this chapter.
\begin{theorem}\label{main_thm_ACN}
Algorithm~\ref{algo_ACN} gives a $3-$approximation.
\end{theorem}

Lastly, we show that the analysis of Algorithm~\ref{algo_ACN} is tight. In particular, we give an instance such that Algorithm~\ref{algo_ACN} has expected cost at least three times the cost of an optimal solution. Let $G_{CC}=(V,E^+,E^-)$ be a complete unweighted graph and $u,v,\in V$ be two special vertices such that $E^-=\{(u,v)\}$ and $E^+=V\times V\setminus E^-.$
\begin{observation}
The expected cost of Algorithm~\ref{algo_ACN} on $G_{CC}$ is at least $3m_{\OPT}$.
\end{observation}
\begin{proof}
Observe that $m_{\OPT}\geq 1$ since $(u,v,w)$ is a bad triangle for any $w\in V\setminus \{u,v\}$. And clustering all the vertices together, $\mathcal{C}=V$, obtains this lower bound. Thus $m_{\OPT}=1$. Now consider the very first iteration of Algorithm~\ref{algo_ACN}. If Algorithm~\ref{algo_ACN} picks any vertex other than $u$ or $v$ as a pivot then the clustering it outputs is $\mathcal{C}=V$ and probability of this event is $\frac{n-2}{n}.$ On the other hand, if it picks $u$ or $v$ as a pivot then the clustering it outputs is $\mathcal{C}=(V\setminus \{a\}, \{a\})$ where $a$ is the vertex in $\{u,v\}$ other than the pivot. The cost of such a clustering is $n-2$ and probability of this event is $\frac{2}{n}.$ Thus,
\begin{align*}
    \mathbb{E}[\ALG]=\frac{n-2}{n}\cdot 1+\frac{2}{n}\cdot(n-2)=\frac{3(n-2)}{n}=3-\frac{6}{n}
\end{align*}
and this expression has limit $3$ as $n$ increases. The statement follows.
\end{proof}

%% file: Chawla-Makarychev-Schramm-Yaroslavtsev-Algorithm.tex
\chapter{Chawla-Makarychev-Schramm-Yaroslavtsev Algorithm}\label{CMSY}
In this chapter we present the $2.06-$approximation algorithm for the Correlation Clustering problem by~\citet*{CMSY15}. It is based on the standard LP relaxation described in Chapter~\ref{CGW}. Its approximation guarantee is almost matching the integrality gap of $2$ given in Section~\ref{gap_CGW}.

Now we describe the algorithm. It takes an optimal solution of the standard LP as an input and rounds it to an integral solution. Loosely speaking, it iteratively picks a random pivot and independently adds each vertex to a new cluster with probability based on the LP length of an edge connecting it to the pivot. In particular, for a pivot $p$ it adds $u$ with probability $1-f^{+}(x_{pu})$ if $(p,u)$ is a positive edge, and with probability $1-f^{-}(x_{pu})$ if $(p,u)$ is a negative edge. Here $f^{+}, f^{-}$ are functions which will be defined later. We give a pseudo-code for this algorithm in Algorithm~\ref{algo_CMSY}.

It is important to mention that Algorithm~\ref{algo_CMSY} can be considered as a generalization of the $2.5-$ approximation algorithm for the Correlation Clustering problem by~\citet*{ACN08}. In particular, ~\citet{ACN08} considered both $f^{+}$ and $f^{-}$ to be identity functions, i.e., $f^{+}(x)=f^{-}(x)=x$ for all $x$.

Furthermore, the analysis of Algorithm~\ref{algo_CMSY} follows the general approach proposed by~\citet*{ACN08}. \citet{ACN08} observed that in order to get upper bounds on the approximation
factors of their algorithms, it is sufficient to consider how these algorithms
behave on triangles.
\begin{algorithm}
   \caption{Rounding Algorithm}
   \label{algo_CMSY}
\begin{algorithmic}[1]
   \STATE Let $V_1=V$ be the set of active vertices and $t=1$.
   \WHILE{$V_t\neq \emptyset$}
   \STATE Pick a pivot $p_t\in V_t$ uniformly at random.
   \STATE Create a new cluster $C_t$; add the pivot $p_t$ to $C_t$.
   \FORALL{$v\in V_t\setminus \{p_t\}$}
   \IF{$(p_t,v)\in E^{+}\mbox{ is a positive edge}$}
   \STATE Add $v$ to $C_t$ with probability $(1-f^{+}(x_{p_{t}v}))$ independently of all other vertices.
   \ELSIF{$(p_t,v)\in E^{-}\mbox{ is a negative edge}$}
   \STATE Add $v$ to $C_t$ with probability $(1-f^{-}(x_{p_{t}v}))$ independently of all other vertices.
   \ENDIF
   \ENDFOR
   \STATE Let $V_{t+1}=V_t\setminus C_t$ and $t=t+1$.
   \ENDWHILE
   \RETURN $\mathcal{C}=\{C_1,\dots, C_{t-1}\}$.
\end{algorithmic}
\end{algorithm}
\section{General Approach: Triangle-Based Analysis}
\label{triple_Analysis}
Consider an instance of Correlation Clustering  $G=(V,E^+,E^-)$ on three vertices $u$, $v$, $w$. Suppose that the edges $(u,v)$, $(v,w)$, and $(u,w)$ have
signs $\sigma_{uv}, \sigma_{vw}, \sigma_{uw}\in \{\pm\}$, respectively. We shall call this instance a triangle $(u,v,w)$ and refer to
the vector of signs $\sigma =(\sigma_{vw}, \sigma_{uw}, \sigma_{uv})$ as the signature of the triangle~$(u,v,w)$.

Let us now assign arbitrary lengths $x_{uv}$, $x_{vw}$, and $x_{uw}$ satisfying the triangle inequality
to the edges $(u,v)$, $(v,w)$, and $(u,w)$ and run one iteration of Algorithm~\ref{algo_CMSY} on the triangle $(u,v,w)$:
\begin{algorithm}
   \caption{One iteration of Algorithm~\ref{algo_CMSY} on triangle $(u,v,w)$}
\begin{algorithmic}
   \STATE Pick a random pivot $p\in \{u,v,w\}$.
   \STATE Create a new cluster $C$. Insert $p$ in $C$.
   \FORALL{$a \in \{u,v,w\}\setminus\{p\}$}
   \IF{$(p,a)\in E^{+}\mbox{ is a positive edge}$}
   \STATE Add $a$ to $C$ with probability $(1-f^{+}(x_{pa}))$ independently of all other vertices.
   \ELSIF{$(p,a)\in E^{-}\mbox{ is a negative edge}$}
   \STATE Add $a$ to $C$ with probability $(1-f^{-}(x_{pa}))$ independently of all other vertices.
   \ENDIF
   \ENDFOR
\end{algorithmic}
\end{algorithm}

Observe that a positive edge $(u,v)$ is in disagreement with $C$ if $u\in C$ and $v\notin C$ or $u\notin C$ and $v\in C$. Similarly, a negative edge $(u,v)$ is in disagreement with $C$ if $u,v\in C$. Let $cost(u,v\given w)$ be the probability that the
edge $(u,v)$ is in disagreement with $C$ given that $w$ is the pivot:
$$cost(u,v\given w) =
\begin{cases}
\Pr[u\in C, v\notin C \text{ or } u\notin C, v\in C\given p = w],& \text{if } \sigma_{uv} = \PlusSign;\\
\Pr[u\in C , v\in C\given p = w],& \text{if } \sigma_{uv} = \MinusSign.
\end{cases}$$

Similarly, let $lp(u,v\given w)$ be the probability that one or both of the vertices $u$ and $v$ are in $C$ given that $w$ is the pivot:
$$lp(u,v\given w) =
\begin{cases}
x_{uv}\cdot \Pr[u\in C \text{ or } v\in C \given p = w],& \text{if } \sigma_{uv} = \PlusSign;\\
(1-x_{uv}) \cdot \Pr[u\in C \text{ or } v\in C \given p = w],& \text{if } \sigma_{uv} = \MinusSign.
\end{cases}$$
It follows that
\begin{align}\label{cost_function_CMSY}
 cost(u,v\given w) =
\begin{cases}
f^{\sigma_{wu}}(x_{wu})+f^{\sigma_{wv}}(x_{wv})-2f^{\sigma_{wu}}(x_{wu})f^{\sigma_{wv}}(x_{wv}),& \text{if } \sigma_{uv} = \PlusSign;\\
(1-f^{\sigma_{wu}}(x_{wu}))(1-f^{\sigma_{wv}}(x_{wv})),& \text{if } \sigma_{uv} = \MinusSign.
\end{cases}
\end{align}
and
\begin{align}\label{lp_function_CMSY}
lp(u,v\given w) =
\begin{cases}
x_{uv}(1-f^{\sigma_{wu}}(x_{wu})f^{\sigma_{wv}}(x_{wv})),& \text{if } \sigma_{uv} = \PlusSign;\\
(1-x_{uv})(1-f^{\sigma_{wu}}(x_{wu})f^{\sigma_{wv}}(x_{wv})),& \text{if } \sigma_{uv} = \MinusSign.
\end{cases}  
\end{align}

We define two functions $ALG^{\sigma}(x,y,z)$ and $LP^{\sigma}(x,y,z)$. To this end,
construct a triangle $(u,v,w)$ with signature $\sigma$ edge lengths $x,y,z$ (where
$x_{vw} = x$, $x_{uw} = y$, $x_{uv} = z$).
%%See Figure~\ref{fig:triangles}.
Then, define
\begin{align}\label{triangle_algo_CMSY}
ALG^{\sigma}(x,y,z)= cost(u,v\given w) +  cost(u,w\given v)+  cost(v,w\given u);
\end{align}
and
\begin{align}\label{triangle_lp_CMSY}
    LP^{\sigma}(x,y,z)= lp(u,v\given w) + lp(u,w\given v)+lp(v,w\given u).
\end{align}

Now we show that in order to upper bound the cost of Algorithm~\ref{algo_CMSY} it is sufficient to analyze its performance on triangles. This observation is first used by~\citet*{ACN08} and explicitly stated by~\citet*{CMSY15}.
\begin{lemma}\label{main_lemma_CMSY}
Consider functions $f^+,f^-$ with $f^+(0)=f^-(0) = 0$. If for all signatures $\sigma=(\sigma_1,\sigma_2,\sigma_3)$
(where each $\sigma_i\in \{\pm\}$) and edge lengths $x$, $y$, $z$ satisfying the triangle inequality,
we have $ALG^{\sigma}(x,y,z)\leq \rho LP^{\sigma}(x,y,z)$, then the approximation factor of Algorithm~\ref{algo_CMSY} is at most $\rho$.
\end{lemma}
\begin{proof}
Our first task is to express the number of mistakes made by Algorithm~\ref{algo_CMSY} and the LP weight in terms of $ALG^\sigma(\cdot)$ and $LP^\sigma(\cdot)$, respectively. In order to do this, we consider the number of mistakes made by the algorithm at each step.

Consider step $t$ of the algorithm. Let $V_t$ denote the set of active (yet unclustered) vertices at the start of step $t$. Let $w\in V_t$ denote the pivot chosen at step $t$.
The algorithm chooses a set $C_t \subseteq V_t$ as a cluster and removes it from the graph. Notice that
for each $u \in C_t$, the constraint imposed by each edge of type $(u,v) \in E^+ \cup E^-$ is satisfied or violated right after step $t$. Specifically, if $(u,v)$ is a positive edge, then the constraint $(u,v)$ is violated if exactly one of the vertices $u,v$ is in $C_t$. If $(u,v)$ is a negative constraint, then it is violated if both $u,v$ are in $C_t$. Denote the number of mistakes at step $t$ by $ALG_t$. Thus,
\begin{align*}
ALG_t=\sum\limits_{\substack{(u,v) \in E^+\\u,v \in V_t}}\ONE \left(u \in C_t, v \not\in C_t \mbox{ or }u \not\in C_t, v \in C_t\right)+\sum\limits_{\substack{(u,v) \in E^-\\u,v \in V_t}}\ONE \left(u \in C_t, v\in C_t\right).
\end{align*}

Similarly, we can quantify the LP weight removed by the algorithm at step $t$, which we denote by $LP_t$.
We count the contribution of all edges $(u,v) \in E^+ \cup E^-$ such that $u \in C_t$ or $v \in C_t$. Thus,
\begin{align*}
LP_t = \sum_{\substack{(u,v) \in E^+\\u,v \in V_t}} x_{uv} \cdot \ONE(u \in C_t \text{ or } v \in C_t)+ \sum_{\substack{(u,v) \in E^-\\u,v \in V_t}} (1 - x_{uv}) \cdot \ONE(u \in C_t \text{ or } v \in C_t).
\end{align*}

Let $\ALG$ denote the cost of the algorithm. Note that the cost of the the algorithm is the sum of the mistakes across all steps, that is $\ALG = \sum_t ALG_t$. Moreover, as every edge is removed exactly once from the graph, we can see that $\LP = \sum_t LP_t$. We will charge the number of the mistakes at step $t$, $ALG_t$, to the LP weight removed at step $t$, $LP_t$. Hence, if we show that $\mathbb{E}[ALG_t] \leq \rho \mathbb{E}[LP_t]$ for every step $t$, then we can conclude that the approximation factor of the algorithm is at most $\rho$, since
$$
\mathbb{E}[\ALG] = \mathbb{E}\bigg[\sum_t ALG_t \bigg] \leq \rho \cdot \mathbb{E}\bigg[\sum_t LP_t \bigg] =  \rho\cdot \LP.
$$

We now express $ALG_t$ and $LP_t$ in terms of $cost(\cdot)$ and $lp(\cdot)$ which are defined in Section~\ref{triple_Analysis}. This will allow us to group together the terms for each triplet $u,v,w$ in the set of active vertices and thus write $ALG_t$ and $LP_t$ in terms of $ALG^\sigma(\cdot)$ and $LP^\sigma(\cdot)$, respectively.

For analysis, we assume that for each vertex $u \in V$, there is a positive (similar) self-loop $(u,u)$.
This edge is never in disagreement. Then, we can define $cost(u,u\given w)$ and $lp(u,u \given w)$ formally as follows:
$cost(u,u \given w) = \Pr[u \in C, u \not\in C \given p = w]
= 0$ and $lp(u,u \given w) = x_{uu} \cdot \Pr[u \in C\given p = w] = 0$ (recall that $x_{uu}=0$ and $f(0)=0$). Observe that
\begin{align}\label{ALGcost}
\mathbb{E}[ALG_t \given  V_t] = \sum_{\substack{(u,v) \in E\\u,v \in V_t}} \bigg( \frac{1}{|V_t|} \sum_{w \in V_t} cost(u,v\given w)\bigg)= \frac{1}{2|V_t|}\sum_{\substack{u,v,w \in V_t \\ u \neq v}} cost(u,v\given w)
\end{align}
and
\begin{align}\label{LPcost}
\mathbb{E}[LP_t \given V_t] = \sum_{\substack{(u,v) \in E\\u,v \in V_t}} \bigg( \frac{1}{|V_t|} \sum_{w \in V_t} lp(u,v\given w)\bigg)= \frac{1}{2|V_t|}\sum_{\substack{u,v,w \in V_t \\ u \neq v}} lp(u,v\given w).
\end{align}

We divide the expressions on the right hand side by $2$ because the terms $cost(u,v \given w)$ and $lp(u,v \given w)$ are counted twice. Now adding the contribution of terms $cost(u,u \given w)$ and $lp(u,u \given w)$ (both equal to $0$) to~(\ref{ALGcost}) and~(\ref{LPcost}), respectively and grouping the terms containing $u,v$ and $w$ together, we get,
\begin{align*}
\mathbb{E}[ALG_t \given V_t] &= \frac{1}{6|V_t|}\sum_{u,v,w \in V_t}\bigg( cost(u,v\given w)+ cost(u,w\given v) + cost(w,v\given u)\bigg)\\
&=\frac{1}{6|V_t|}\sum_{u,v,w \in V_t} ALG^\sigma(x,y,z);
\end{align*}
and
\begin{align*}
\mathbb{E}[LP_t \given V_t] =& \frac{1}{6|V_t|}\sum_{u,v,w \in V_t}\bigg(lp(u,v\given w) + lp(u,w\given v) + lp(w,v\given u)\bigg)\\
=&\frac{1}{6|V_t|}\sum_{u,v,w \in V_t}LP^\sigma(x,y,z).
\end{align*}

Thus, if $ALG^\sigma(x,y,z) \leq \rho LP^\sigma(x,y,z)$ for all signatures and edge lengths $x,y,z$ satisfying the triangle inequality, then $\mathbb{E}[ALG_t \given V_t] \leq \rho \cdot \mathbb{E}[LP_t \given V_t]$, and, hence, $\mathbb{E}[\ALG] \leq \rho \cdot\LP$ which finishes the proof.
\end{proof}

Now we state the main theorem of this chapter.
\begin{theorem}\label{main_thm_CMSY}
Algorithm~\ref{algo_CMSY} with rounding functions
\begin{align*}
 f^{+}(x) =
 \begin{cases}
 0, & \mbox{if } x<a \\ \big(\frac{x-a}{b-a}\big)^2, & \mbox{if } x\in [a,b],\\
 1, & \mbox{if } x\geq b
 \end{cases}
 \;\;\;\;\;\;f^{-}(x)=x,
 \end{align*}
gives a $(2.06-\varepsilon)-$approximation for $a=0.19$ and $b=0.5095$, and a constant $\varepsilon$ with $0<\varepsilon<0.01$.
\end{theorem}
\noindent
\textit{Proof Sketch:} The analysis of Algorithm~\ref{algo_CMSY} relies on Lemma~\ref{main_lemma_CMSY}. It is sufficient to show that
\begin{equation*}\label{triangle_inequality_CMSY}
ALG^{\sigma}(x,y,z)\leq \rho\cdot LP^{\sigma}(x,y,z),
\end{equation*}
holds for every triangle $(u,v,w)$ with edge lengths $(x,y,z)$ (satisfying the triangle inequality) and signature $\sigma = (\sigma_{vw},\sigma_{uw}, \sigma_{uv})$ where $\rho=2.06$.
To simplify the exposition, here we demonstrate only the signature $\sigma=(\PlusSign,\MinusSign,\MinusSign)$.\footnote{The proof is technical and involves multiple case analysis. See~\citet{CMSY15} (Appendix A) for the complete proof.}

Before delving into the details, we describe the general proof strategy and give an intuition for the choice of $f^+$ and $f^-$.
The main idea is to represent $\rho\cdot LP^{\sigma}(x,y,z)-ALG^{\sigma}(x,y,z) $ as a polynomial $ \mathcal{P}(x,y,z,f^{\sigma_{vw}}(x),f^{\sigma_{uw}}(y),f^{\sigma_{uv}}(z)) $ and show that it is nonnegative for all possible edge lengths $(x,y,z)$ satisfying the triangle inequality and for all possible signatures $\sigma$. Observe that we can compute $\mathcal{P}$ by~(\ref{cost_function_CMSY}),~(\ref{lp_function_CMSY}),~(\ref{triangle_algo_CMSY}) and~(\ref{triangle_lp_CMSY}). Furthermore, it can be shown that it is sufficient to analyze $\mathcal{P}$ for $(x,y,z)$ such that triangle inequality is tight. Now we give an informal explanation for the choice of $f^{+}$ and $f^{-}$.

For $\sigma=(\PlusSign,\MinusSign,\MinusSign)$ and edge lengths $(0,x,x)$, $ \mathcal{P}\geq 0$ implies
\begin{align*}
    f^{-}(x)\geq\sqrt{1-\rho(1-x)}
\end{align*}
(recall that $f^-(0)=f^+(0)=0$ must hold). The function $f^-(x)=x$ satisfies the above condition for $\rho=2.06.$ Thus we take $f^-(x)=x$, as this choice is an easy candidate for the analysis.

Similarly, for $\sigma=(\PlusSign,\PlusSign,\MinusSign)$ and edge lengths $(x,x,2x)$,
\begin{align*}
    \mathcal{P}=-1+\rho-4x-4x(-2+\rho x)f^+(x)-(1+\rho-2\rho x)f^+(x)^2.
\end{align*}
Solving $\mathcal{P}\geq 0$ in terms of $f^+(x)$ gives us the following lower bound on $f^{+}(x)$ for $x\in[0,\frac{1}{2}]$
\begin{align*}
    f^+(x)\geq \frac{8x-4\rho x^2-\sqrt{(4\rho x^2-8x)^2-4(1-\rho+4x)(1+\rho -2\rho x)}}{2(1+\rho-2\rho x)}
\end{align*}

Finally, for $\sigma=(\PlusSign,\PlusSign,\PlusSign)$ and edge lengths $(x,x,0)$, $\mathcal{P}\geq 0$ implies the following upper bound on $f^{+}$ for $x\in[0,1]$
\begin{align*}
    f^+(x)\leq 1-\sqrt{1-\rho x}.
\end{align*}
Furthermore, it can be shown that $f^+,f^-$ defined in Theorem~\ref{main_thm_CMSY} satisfy the above conditions for $\rho=2.06$, $a=0.19$ and $b=0.5095$. Now we give a complete analysis of the signature $\sigma=(\PlusSign,\MinusSign,\MinusSign)$.
\begin{lemma}
Let $\sigma=(\PlusSign,\MinusSign,\MinusSign)$. Then $2\cdot LP^{\sigma}(x,y,z)\geq ALG^{\sigma}(x,y,z) $ for all possible edge lengths $(x,y,z)$ satisfying the triangle inequality.
\end{lemma}
\begin{proof}
Due to~(\ref{cost_function_CMSY}) and~(\ref{lp_function_CMSY}) we have

\begin{align*}
 cost(u,v\given w)&=(1-f^-(y))(1-f^+(x)), & lp(u,v\given w)&=(1-z)(1-f^-(y)f^+(x));\\
 cost(u,w\given v)&=(1-f^-(z))(1-f^+(x)), & lp(u,w\given v)&=(1-y)(1-f^-(z)f^+(x));\\
 cost(v,w\given u)&=f^-(y)+f^-(z)(1-2f^-(y)), & lp(v,w\given u)&=x(1-f^-(y)f^-(z)).
\end{align*}
Since $f^-(y)=y$ and $f^-(z)=z$,~(\ref{triangle_algo_CMSY}) and~(\ref{triangle_lp_CMSY}) imply
\begin{align*}
    ALG^{\sigma}(x,y,z)&=2-2yz-2f^+(x)+f^+(x)y+f^+(x)z,\\
    LP^{\sigma}(x,y,z)&=2+x-y-z-xyz-f^+(x)y-f^+(x)z+2f^+(x)yz.
\end{align*}
Then
\begin{align*}
   2\cdot LP^{\sigma}(x,y,z)- ALG^{\sigma}(x,y,z) =& 2(2+x-y-z-xyz-f^+(x)y-f^+(x)z+2f^+(x)yz)\\
   &-(2-2yz-2f^+(x)+f^+(x)y+f^+(x)z)\\
   =&2(1-y)(1-z)+2x(1-yz)+f^+(x)(2-3z-3y+4yz)\\
   \geq&2(1-y)(1-z)+3f^+(x)(1-z-y+yz)\\
   =&(1-y)(1-z)(2+3f^+(x))\geq 0
\end{align*}
where the first inequality follows from the assumption $2x\geq f^+(x)$. We note that the function $f^+$ is bounded above by $2x$ on $[0,1]$. In particular, for $x<a$ and $x> b$ this is clear immediately, and for $ x\in[a,b]$ we note that $f^+$ is convex and $f^+(a)\leq 2a$ and $f^+(b)\leq 2b$. 
\end{proof}